\renewcommand{\theequation}{\arabic{section}.\arabic{equation}}
\documentclass{ijsipaper}
\usepackage{graphicx,amsmath,amssymb,bbm}
\def\id{\mathbbm I}

\newcommand{\Tr}{{\rm Tr}\,}
\newcommand{\ort}{{\scriptscriptstyle \perp}}

\newcommand{\smb}{{\scriptscriptstyle B}}

\newcommand{\sma}{{\scriptscriptstyle A}}
\newcommand{\smc}{{\scriptscriptstyle C}}
\newcommand{\smab}{{\scriptscriptstyle AB}}

\newcommand{\ket}[1]{\vert #1 \rangle}

\ifprodtf \else \usepackage{latexsym}\fi
\newcommand\black{\ensuremath{\blacktriangleright}}
\newcommand\white{\ensuremath{\vartriangleright}}
\newif\ifamsfontsloaded
\ifprodtf
  \newcommand\whbl{\white\kern-.1em--\kern-.1em\black}
  \newcommand\blwh{\black\kern-.1em--\kern-.1em\white}
  \newcommand\blbl{\black\kern-.1em--\kern-.1em\black}
  \newcommand\whwh{\white\kern-.1em--\kern-.1em\white}
  \amsfontsloadedtrue
\else
  \checkfont{msam10}
  \iffontfound
    \IfFileExists{amssymb.sty}
      {\usepackage{amssymb}\amsfontsloadedtrue
       \newcommand\whbl{\white\kern-.125em--\kern-.125em\black}%
       \newcommand\blwh{\black\kern-.125em--\kern-.125em\white}%
       \newcommand\blbl{\black\kern-.125em--\kern-.125em\black}%
       \newcommand\whwh{\white\kern-.125em--\kern-.125em\white}}
      {}
  \fi
\fi

\newtheorem{theorem}{Theorem}[section]
\sponsor{The italian MIUR by the project FIRB-LiCHIS-RBFR10YQ3H}
\begin{document}
\title{Quantum ensembles and the statistical operator: a tutorial}
\author[Yinxi\`u Zhan and Matteo G. A. Paris]
{Yinxi\`u Zhan and Matteo G. A. Paris\vspace*{2mm}\\
Dipartimento di Fisica dell'Universit\`a degli Studi di Milano, I-20133
Milano, Italia.  }
\correspond{Email: {\tt matteo.paris@fisica.unimi.it}}
\receiveddate{} %
\reviseddate{} %
\accepteddate{} %
\onlinepulishdate{} %
\paperurl{\bf} %
\setcounter{page}{1}
\label{firstpage}
\maketitle
\makesponsor %
\makecorrespond %
\makeusefuldate %
\begin{abstract}
The main purpose of this tutorial is to elucidate in details what should
be meant by ensemble of states in quantum mechanics, and to properly address 
the problem of discriminating, exactly or approximately, two
different ensembles. To this aim we review the notion and the definition
of {\em quantum ensemble} as well as its relationships with the concept
of {\em statistical operator} in quantum mechanics. We point out the
implicit assumptions contained in introducing a correspondence between
quantum ensembles and the corresponding single-particle statistical
operator, and discuss some issues arising when these assumptions are not
satisfied. We review some subtleties leading to apparent paradoxes, 
and illustrate the role of approximate quantum cloning. In particular, 
we review some examples of practical interest where different (but 
equivalent) preparations of a quantum system, i.e.  different ensembles 
corresponding to the same single-particle
statistical operator, may be successfully discriminated exploiting
multiparticle correlations, or some {\em a priori} knowledge about 
the number of particles in the ensemble.
\end{abstract}
\begin{keywords}
Quantum ensembles; statistical operator; density matrix; no-cloning
theorem.
\end{keywords}
\makepaperinfo\normalsize\parindent=7mm
\section{Introduction}
As matter of fact, fundamental postulates of quantum mechanics do not
allow to discriminate two ensembles corresponding the same statistical 
(density) operator. Approximate discrimination is also forbidden by 
fundamental laws of physics, since it would violate the no-signaling
condition imposed by causality. 
Despite this facts, one may find some claims in the literature about the
possible discrimination of ensembles corresponding to the same density
operator. In particular, specific schemes have been put 
forward \cite{ghi75,ghi76,per93,esp01}, 
involving finite set of particles sampled from a 
given ensemble, or ensembles prepared with inner correlations. 
\par
The main goal of this tutorial is to remove the apparent paradoxes by 
addressing in details what should
be meant by ensemble of states in quantum mechanics, and to address
properly the problem of discriminating, exactly or approximately, two 
different ensembles. To this purpose, we review the notion and the definition 
of {\em quantum ensemble} as well as its relationships with the concept 
of {\em statistical operator} in quantum mechanics. We point out the 
implicit assumptions contained in introducing a correspondence between 
quantum ensembles and the corresponding single-particle statistical 
operator, and review the apparent paradoxes arising 
when these assumptions are not satisfied. 
\par
As we will see, the paradoxes arise from mixing up the two different
concepts of single-particle and many-particles statistical operators or
to assume a fixed number of particle in the ensembles
\cite{new76,tol06,lon06}.  A detailed analysis of the measurement
schemes mentioned above shows that discrimination is indeed possible,
but also that the involved ensembles correspond to the same
single-particle density operator but different many-particles ones.  In
other words, there are no paradoxes unless the measurement schemes are
analyzed in naive way.
\par
The tutorial is structured as follows. In the next Section we review 
the definition of ensemble and statistical operator in quantum
mechanics, also emphasizing the implicit assumptions needed to 
establish a correspondence between 
quantum ensembles and the corresponding single-particle statistical 
operator. In Section \ref{s:qclo} we briefly discuss 
the connection between the no-cloning theorem and the impossibility of
discriminating ensembles with the same density operator and show
that also approximate approximate quantum cloning machines (AQCM)
cannot be used for this task. In Section \ref{s:exa} we discuss some
measurement schemes where discrimination of seemingly equivalent
ensembles is realized and show that this situation arise when the
implicit assumptions contained in introducing a correspondence between
quantum ensembles and the corresponding single-particle statistical
operator, are not satisfied. Section \ref{s:out} closes the tutorial
with some concluding remarks.
\section{The statistical operator of a quantum system}
According to the basic postulates of quantum mechanics, the states of a 
physical system are described by normalized vectors  $|\psi\rangle$,  
$\langle \psi|\psi\rangle=1$,  of a Hilbert space 
$H$. Composite systems, either made by more than one physical object 
or by the different degrees of freedom of the same entity, are described 
by tensor product $H_1 \otimes H_2 \otimes ...$ of the corresponding 
Hilbert spaces and the overall state of the system is a vector in the 
global space. As far as the Hilbert space description of physical
systems is adopted then we have the superposition principle, 
which says that if $|\psi_1\rangle$ and $|\psi_2\rangle$ 
are possible states of a system, then also any (normalized) linear 
combination $\alpha|\psi_1\rangle+\beta|\psi_2\rangle$, $\alpha,\beta\in
{\mathbbm C}$, $|\alpha|^2+|\beta|^2=1$ of the two states is an admissible 
state of the system.
\par
Observable quantities are described by Hermitian operators $X$. Any 
Hermitian operator $X=X^\dag$, admits a spectral decomposition 
$X=\sum_x x P_x$, in terms of its real eigenvalues $x$, which are
the possible value of the observable, and of the projectors 
$P_x=|x\rangle\langle x|$, $P_x,P_{x'}=\delta_{xx'}P_x$ on its
eigenvectors $X|x\rangle=x|x\rangle$, which form a basis for the
Hilbert space, i.e. a complete set of orthonormal states with the 
properties $\langle x|x'\rangle = \delta_{xx'}$ (orthonormality), 
and $\sum_x |x\rangle\langle x| =\id$ (completeness, we omitted to
indicate the dimension of the Hilbert space).  $L(H)$ is the linear 
space of (linear) operators from $H$ to $H$, which itself is a Hilbert s
pace with scalar product provided by the
trace operation, i.e. upon denoting by $|A\rangle\rangle$ operators seen
as elements of $L(H)$, we have $\langle\langle A| B\rangle\rangle =
\hbox{Tr}[A^\dag B]$.
\par
The probability $p_x$ of obtaining the outcome $x$ from the measurement of 
the observable $X$ and the overall expectation value $\langle X\rangle$ 
are given by
\begin{align}
p_x & = \left|\langle\psi| x\rangle\right|^2 
= \langle\psi| P_x| \psi \rangle
=\sum_n \langle \psi| \varphi_n\rangle\langle \varphi_n |P_x|\psi\rangle
\notag \\
&=\sum_n 
\langle \varphi_n |P_x|\psi\rangle
\langle \psi| \varphi_n\rangle
=
\hbox{Tr}\left[|\psi\rangle\langle\psi|\, P_x\right]
\end{align}
and $\langle X\rangle = \langle\psi| X|\psi\rangle =
\hbox{Tr}\left[|\psi\rangle\langle\psi|\, X \right]$.
This is the {\em Born rule} and it is the fundamental recipe to connect the
mathematical description of a quantum state to the prediction of quantum 
theory on the results of an experiment. 
The state of the system after 
the measurement is the projection of the state before the measurement on 
the eigenspace of the observed eigenvalue, i.e.
$$
|\psi_x\rangle = \frac{1}{\sqrt{p_x}}\, P_x |\psi\rangle
$$
\par
Let us nos suppose to deal with a quantum system whose preparation is not completely
un\-der con\-trol. What we know is that the system is prepared in the
state $|\psi_k\rangle$ with probability $p_k$, i.e. that the system
is described by the statistical ensemble $\{p_k, |\psi_k\rangle\}$,
$\sum_k p_k=1$, where the states $\{|\psi_k\rangle\}$ are not, in general, 
orthogonal. The expected value of an observable $X$ may be evaluated 
as follows
\begin{align}\notag
\langle X\rangle & = \sum_k p_k \langle X \rangle_k = \sum_k p_k
\langle\psi_k | X| \psi_k \rangle = \sum_{n\,p\,k} p_k
\langle\psi_k|\varphi_n\rangle\langle\varphi_n |
X|\varphi_p\rangle\langle\varphi_p| \psi_k \rangle 
\\ \notag &=
\sum_{n\,p\,k} p_k
\langle\varphi_p| \psi_k \rangle\langle\psi_k|\varphi_n\rangle
\langle\varphi_n |X|\varphi_p\rangle
= \sum_{n\,p} 
\langle\varphi_p| \varrho| \varphi_n\rangle
\langle\varphi_n |X|\varphi_p\rangle
\\ \notag 
&=\sum_{p} 
\langle\varphi_p| \varrho\,X|\varphi_p\rangle
= \hbox{Tr}\left[\varrho\,X\right]\,
\end{align}
where $$\varrho = \sum_k p_k\,|\psi_k\rangle\langle\psi_k |$$
is the {\em statistical (density) operator} of the system 
under investigation.
The $|\varphi_n\rangle$'s in the above formula are a basis for the
Hilbert space and we used the trick of suitably inserting two resolutions of the
identity $\id = \sum_n |\varphi_n\rangle\langle\varphi_n|$.   
The formula is of course trivial if the $|\psi_k\rangle$'s are themselves 
a basis or a subset of a basis.  
\par
\begin{theorem}
An operator $\varrho$ is the density operator 
associated to an ensemble $\{p_k, |\psi_k\rangle\}$ is and only if it 
is a positive (hence selfadjoint) operator with unit 
trace $\hbox{Tr}\left[\varrho\right]=1$. 
\end{theorem}
\begin{proof}
If $\varrho=\sum_k p_k |\psi_k\rangle\langle\psi_k |$
is a density operator then $\hbox{Tr}[\varrho]=\sum_k p_k=1$ and
for any vector $|\varphi\rangle \in H$,
$\langle\varphi|\varrho|\varphi\rangle = \sum_k p_k
|\langle\varphi|\psi_k\rangle|^2 \geq 0$. Viceversa, if 
$\varrho$ is a positive operator with unit trace than it can be
diagonalized and the sum of eigenvalues is equal to one. Thus it can be
naturally associated to an ensemble. 
\end{proof}
As it is true for any operator, the density operator may be expressed in terms
of its matrix elements in a given basis, i.e. $\varrho=\sum_{np}
\varrho_{np} |\varphi_n\rangle\langle\varphi_p|$ where
$\varrho_{np}=\langle\varphi_n|\varrho|\varphi_p\rangle$ is usually
referred to as the {\em density matrix} of the system.
Of course, the density matrix of a state is diagonal if we use
a basis which coincides or includes the set of eigenvectors of the
density operator, whereas it contains off-diagonal elements otherwise.
\par
Different ensembles may lead to the same density operator. In this case 
they have the same expectation values for any operator and thus are physically
indistinguishable. In other words, 
different preparations of ensembles leading to the same density operator
are actually the same state, i.e. the density operator appears to provide 
the natural and fundamental quantum description of physical systems
\cite{esp95,fan57}.
\par
How this reconciles with the above postulate 
saying "physical systems are described by vectors in a Hilbert space"? 
\par
In order to see how it works let us first notice that according to the 
postulates above the action of "measuring nothing" should be described by 
the identity operator $\id$. Indeed the identity is Hermitian and has the 
single eigenvalues $1$, corresponding to the persistent result of measuring 
nothing. Besides, the eigenprojector corresponding to the eigenvalue $1$ is 
the projector over the whole Hilbert space and thus we have the
consistent prediction that the state after the "measurement" is
left unchanged. Let us consider a situation in which a bipartite system
prepared in the state $|\psi_\smab\rangle\rangle \in H_{\sma} \otimes H_{\smb}$
is subjected to the measurement of an observable $X=\sum_x P_x \in L(H_\sma)
$, $P_x=|x\rangle\langle x|$ i.e. a measurement involving only the degree 
of freedom described by the Hilbert space $H_\sma$. The overall observable
measured on the system is thus $\boldsymbol{X}=X\otimes \id_\smb$, with 
spectral decomposition  $\boldsymbol{X}= \sum_x x\, \boldsymbol{Q}_x$, 
$\boldsymbol{Q}_x=P_x\otimes \id_\smb$. The probability distribution 
of the outcomes is then obtained using the Born rule, i.e.
\begin{align}
\label{star}
p_x = \hbox{Tr}_\smab
\Big[|\psi_\smab\rangle\rangle\langle\langle\psi_\smab |\,
P_x \otimes \id_\smb\Big] 
\end{align}
On the other hand, since the measurement has been performed only on the
system $A$ one expects the Born rule to be valid also at the level
of single system and a question arises on the form of the object 
$\varrho_\sma$ which allows one to write  
$p_x = \hbox{Tr}_\sma \left[\varrho_\sma\, P_x\right]$
i.e. the Born rule as a trace only over the Hilbert space $H_\sma$. 
Upon inspecting Eq. (\ref{star}), one sees that a suitable  mapping  
$|\psi_\smab\rangle\rangle\langle\langle\psi_\smab | \rightarrow \varrho_\sma$ is
provided by the partial trace 
$\varrho_\sma=\hbox{Tr}_\smb\left[|\psi_\smab\rangle\rangle\langle
\langle\psi_\smab |\right]$. Indeed, for the operator $\varrho_\sma$ 
defined by the above partial trace we have $\hbox{Tr}_\sma[\varrho_\sma]=
\hbox{Tr}_\smab\left[|\psi_\smab\rangle\rangle\langle
\langle\psi_\smab |\right]=1$ and, for any vector $|\varphi\rangle\in
H_\sma$ , 
$\langle\varphi_\sma|\varrho_\sma|\varphi_\sma\rangle = \hbox{Tr}_\smab
\left[|\psi_\smab\rangle\rangle\langle
\langle\psi_\smab |\,
|\varphi_\sma\rangle\langle\varphi_\sma|\otimes
\id_\smb\right] \geq 0$. Being a positive, unit trace, operator
$\varrho_\sma$ is itself a density operator according to the 
definition above. It should be also noticed that
actually, the partial trace is the unique
operation which allows to maintain the Born rule at both level i.e. the
unique operation leading to the correct description of observable
quantities for subsystems of a composite system. Let us state this as a
theorem \cite{Nie00}: 
\begin{theorem}
The unique mapping 
$|\psi_\smab\rangle\rangle\langle\langle\psi_\smab | \rightarrow \varrho_\sma =
f(\psi_\smab)$ from $H_\sma \otimes H_\smb$ to $H_\sma$ for which  
$\hbox{Tr}_\smab
\left[|\psi_\smab\rangle\rangle\langle\langle\psi_\smab |\,
P_x \otimes \id_\smb\right] = \hbox{Tr}_\sma \left[f(\psi_\smab)\,
P_x\right]$ is the partial trace $f(\psi_\smab)\equiv \varrho_\sma = 
\hbox{Tr}_\smb\left[|\psi_\smab\rangle\rangle\langle
\langle\psi_\smab |\right]$. 
\end{theorem}
\begin{proof} Basically the proof reduces to the fact that the 
set of operators on $H_\sma$ is itself a Hilbert space $L(H_\sma)$ 
with scalar product given by $\langle\langle A| B\rangle\rangle = 
\hbox{Tr}[A^\dag B]$. Indeed, let us consider a basis of operators
$\{M_k\}$ for $L(H_\sma)$ and expand $f(\psi_\smab) =\sum_k M_k
\hbox{Tr}_\sma[M_k^\dag f(\psi_\smab)]$. Since $f$ is the map to
preserve the Born rule we have 
$$
f(\psi_\smab) =\sum_k M_k
\hbox{Tr}_\sma[M_k^\dag\, f(\psi_\smab)]
= \sum_k M_k
\hbox{Tr}_\smab\left[M_k^\dag\otimes\id_\smb\,
|\psi_\smab\rangle\rangle\langle\langle\psi_\smab |\right]\,
$$
and the thesis follows from the fact that in a Hilbert space the
decomposition on a basis is unique. 
\end{proof}
The above result can be easily generalized to the case of a system 
which is initially described by a density operator $\varrho_\smab$ 
and thus we conclude that when we focus attention to a subsystem of 
a composite larger system the unique mathematical description of 
ignoring part of the degrees of freedom is provided by the partial trace. 
It remains to be proved that the partial trace of a density operator 
is a density operator too. This is a very consequence of the definition
that we put in form a little theorem. 
\begin{theorem} The partial traces 
$\varrho_\sma = \hbox{Tr}_\smb[\varrho_\smab]$, 
$\varrho_\smb = \hbox{Tr}_\sma[\varrho_\smab]$ 
of the density operator of a bipartite system are
themselves density operators for the reduced systems.
\end{theorem}
\begin{proof} We have 
$\hbox{Tr}_\sma [\varrho_\sma] = \hbox{Tr}_\smb [\varrho_\smb] =
\hbox{Tr}_\smab[\varrho_\smab]=1$ 
and, for any state 
$|\varphi_\sma\rangle\in H_\sma$, $|\varphi_\smb\rangle\in H_\smb$,
\begin{align}
\langle\varphi_\sma|\varrho_\sma|\varphi_\sma\rangle &= \hbox{Tr}_\smab
\left[\varrho_\smab\, |\varphi_\sma\rangle\langle\varphi_\sma|\otimes
\id_\smb\right] \geq 0 \notag \\
\langle\varphi_\smb|\varrho_\smb|\varphi_\smb\rangle &= \hbox{Tr}_\smab
\left[\varrho_\smab\,\id_\sma \otimes
|\varphi_\smb\rangle\langle\varphi_\smb|
\right] \geq 0 \notag\,, 
\end{align}
which prove  positivity of both $\varrho_\sma$ and $\varrho_\smb$.
\end{proof}
It also follows that the state of the "unmeasured" subsystem after 
the observation of a specific outcome may be obtained as a partial
trace of the projection of the state before the measurement on the 
eigenspace of the observed eigenvalue, i.e.
$$
\varrho_{\smb x} = \frac{1}{p_x} \hbox{Tr}_\sma\left[ P_x\otimes\id_\smb
\,\varrho_\smab\, P_x\otimes\id_\smb \right]
= \frac{1}{p_x} \hbox{Tr}_\sma\left[\varrho_\smab\, P_x\otimes\id_\smb
\right]\,
$$
where, to write the second equality, we made use of the circularity of
the trace and of the fact that we are in presence of a factorized
projector. This is also referred to as the "conditional state" of
system $B$ after the observation of the outcome $x$ from a measurement 
of the observable $X$ performed on the system $A$.
\subsection{Discussion}
In several textbooks a distinction is made between ensembles coming from
the ignorance about the preparation of a system and those emerging from
measurements performed on bipartite systems. They are usually referred to
as ensembles of the {\em proper} and {\em improper} kind respectively
\cite{coh99}.
Actually, as it emerges clearly from the derivation reported above,
there is no fundamental difference between the two kinds of ensembles
and this classification is somehow artificial (though it has been useful 
in the development of the field).
\par
The emerging definition is the following: a {\em quantum ensemble} is a 
collection of repeated identical preparations of the system, randomly 
generated according to a given probability distribution. When this
definition is applicable then we have the fundamental result reported 
above: two ensembles corresponding to the same statistical operator
cannot be discriminated by any kind of measurement, i.e. they are
physically indistinguishable and do not correspond to different 
physical entities \cite{par12}. 
\par
We want to emphasize, however, that the above definition contains two implicit
assumptions that may not be verified in all the physical situation of
interest. They are: i) the preparations are identical and random, i.e.
no correlations are present within the ensemble, e.g. between
subsequent preparations; ii) the number of preparations is not
known or fixed: strictly speaking statistical operators describe ensembles
made of infinite preparations, though this should be intended as {\em 
large enough} for the advent of the law of large numbers. 
Whenever the two conditions are not fulfilled the correspondence between 
ensembles and (single-particle) 
statistical operators is no longer ensured and (apparent)
paradoxical situations may arises. 
\section{Quantum cloning and discrimination of ensembles}\label{s:qclo}
Before addressing the issues arising when the above assumptions are not
satisfied, we briefly review and discuss the connections between i) the
impossibility of discriminating ensembles with the same density operator
and ii) the impossibility of perfectly replicating quantum information, i.e. the
so-called no-cloning theorem. We also show that even approximate 
discrimination is not possible since this would violate the no-signaling 
condition imposed by causality. Quite obviously, feasible (approximate) 
quantum cloning machines, which have been designed to fulfill 
this condition, cannot be employed as well for discriminating ensembles 
\cite{woo82,buz01,ant05}.
\par
Let us start by reviewing the no-cloning theorem in its general form
\begin{theorem}
There is no unitary operation $U$ on $H_\sma\otimes 
H_\smb\otimes H_\smc$ that for given $\ket{\omega}_\smb$ and
$\ket{A}_\smc$ 
is able to implement the transformation $U\ket{\psi}_\sma
\otimes\ket{\omega}_\smb \otimes
\ket{A}_\smc = 
\ket{\psi}_\sma\otimes\ket{\psi}_\smb\otimes\ket{A_\psi}_\smc$ 
for any $\ket{\psi}_\sma$.
\end{theorem}
From the operational point of view the theorem says that no
physical device (initially prepared in the state 
$|A\rangle\in H_\smc$) may obtain two copies
of a generic quantum state $|\psi\rangle_\sma$ starting from 
a single copy and by coupling the system under investigation to 
an ancillary system $B$ of the same dimension.
\begin{proof}
The proof is based on the sole request of linearity of quantum mechanics.
In fact if we require the device to work for a pair of states
$|\varphi_0\rangle$  and $|\varphi_1\rangle$, i.e
\begin{align}
U\, \ket{\varphi_0}_\sma \otimes \ket{\omega}_\smb \otimes \ket{A}_\smc
&=\ket{\varphi_0}_\sma \otimes \ket{\varphi_0}_\smb \otimes \ket{A_0}_\smc
\notag \\
U\, \ket{\varphi_1}_\sma \otimes \ket{\omega}_\smb \otimes \ket{A}_\smc
&=\ket{\varphi_1}_\sma \otimes \ket{\varphi_1}_\smb \otimes
\ket{A_1}_\smc\,,\notag
\end{align}
then, by linearity, one has 
$$
U\, \left(\ket{\varphi_0}_\sma  + \ket{\varphi_1}_\sma\right) 
\otimes \ket{\omega}_\smb \otimes \ket{A}_\smc
= \ket{\varphi_0}_\sma \otimes \ket{\varphi_0}_\smb \otimes \ket{A_0}_\smc
+ \ket{\varphi_1}_\sma \otimes \ket{\varphi_1}_\smb \otimes
\ket{A_1}_\smc\,,
$$
which is not what we are expecting from a cloning device, since 
the cloning of a superposition should correspond to
$$
U\, 
\frac1{\sqrt{2}}\left(\ket{\varphi_0}_\sma  + \ket{\varphi_1}_\sma\right) 
\otimes \ket{\omega}_\smb \otimes \ket{A}_\smc
=
\frac1{\sqrt{2}}\left(\ket{\varphi_0}_\sma  + \ket{\varphi_1}_\sma\right) 
\otimes
\frac1{\sqrt{2}}\left(\ket{\varphi_0}_\smb  + \ket{\varphi_1}_\smb\right) 
\otimes \ket{A_{01}}_\smc\,.$$
In other words, linearity of quantum mechanics forbids the existence of 
a cloning machine for any unitary on $H_\sma\otimes  H_\smb\otimes H_\smc$, 
i.e. any map on $H_\sma\otimes H_\smb$.
\end{proof}
\begin{figure}
\includegraphics[width=0.5\textwidth]{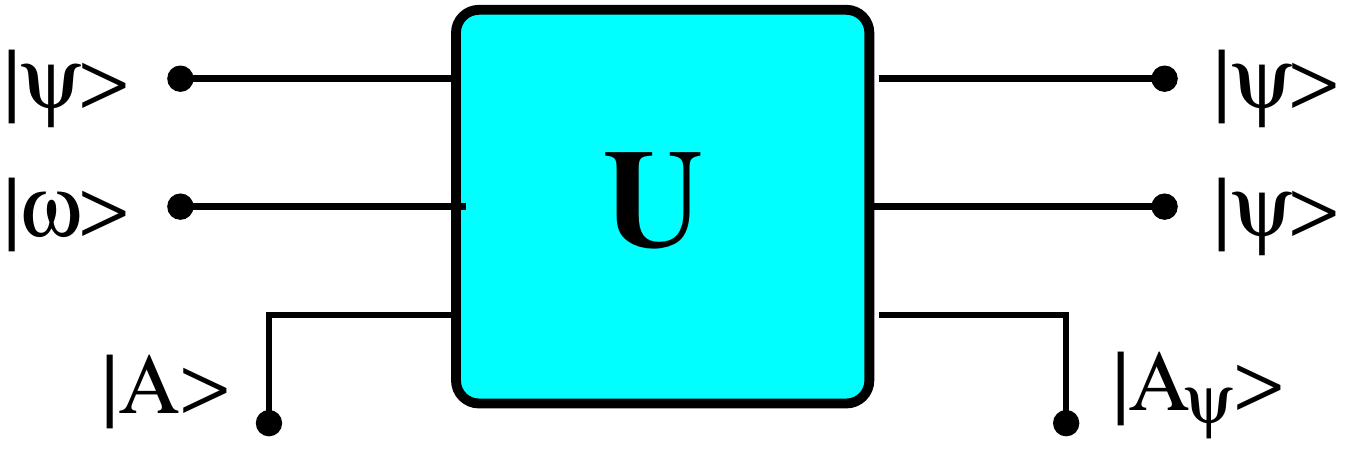}
\label{f:cl2}
\caption{Schematic diagram of a hypothetical cloning device including the 
degrees of freedom of the apparatus. Linearity of quantum mechanics
forbids the existence of this kind of device.
} 
\end{figure}
\par 
It is well known from the discussion about Bell's inequalities 
that quantum nonlocality cannot be used to implement any kind
of superluminal communication and, in turn, to violate
causality (no-signaling condition), even in conjunction with the
reduction postulate \cite{ant05}.  
Let us now give some more details in order to connect this fact 
with the impossibility of quantum cloning and of discrimination 
of ensembles corresponding to the same statistical operator. 
Indeed, it is known as the no-cloning
theorem was triggered as a response to a (wrong) proposal for
superluminal communication named FLASH (first light amplification
superluminal hookup) put forward by N. Herbert in 1982 \cite{her82}. 
\par
The argument goes as follows: Let us assume that Alice and Bob share an 
entangled state of the 
form $|\psi\rangle\rangle=\frac1{\sqrt{2}}\left(|00\rangle\rangle + 
|11\rangle\rangle\right)$, where $|jj\rangle\rangle=|j\rangle\otimes|j\rangle$
and we employ the standard basis made of eigenstates of $\sigma_3$ in both
Hilbert spaces. If Alice performs a
measurement of the spin in a generic direction $\sigma_\phi$ 
she may obtain one of the two possible
outcomes $\pm 1$ with equal probability  $p=\frac12$ and, correspondingly,
Bob's state is projected onto one of the two possible conditional states
$|0\rangle_\phi=\cos\phi|0\rangle+\sin\phi |1\rangle$ and 
$|1\rangle_\phi=\cos\phi|0\rangle+\sin\phi |1\rangle$. Of course, 
if Bob does not know 
the result of Alice' measurement his conditional state is given by
$\varrho^{(1)}_\smb = \frac12 \left(
|0\rangle_{\phi\phi}\langle 0| + |1\rangle_{\phi\phi}\langle 1|\right) = \frac12
\id$. Analogously, if Alice performs the measurement of $\sigma_3$, then
the reduction occurs on the states $|0\rangle$ and $|1\rangle$ but, 
without the knowledge of Alice's results, the overall 
conditional state of Bob is $\varrho^{(3)}_\smb = \frac12 \left(
|0\rangle\langle 0| + |1\rangle\langle 1|\right) = \frac12
\id$. Being $\varrho^{(3)}_\smb=\varrho^{(1)}_\smb$ the impossibility 
of discriminating the two ensembles is equivalent to the impossibility
for Bob to infer which measurement has been performed by Alice, i.e. 
it is not possible to exploit nonlocal correlations of entangled states 
to transmit information. The same argument may be easily repeated for any 
choice of the pair of measurements performed by Alice. In order to make
the two ensembles distinguishable, at least partially, Alice should send
to Bob some piece of information about the results of her measurement, 
using some traditional communication channel, thus
"saving causality".
\par
Let us now assume that Bob has at disposal a perfect quantum cloning machine
and use it, in a scheme like the right part of Fig. \ref{f:cl3},
whenever Alice performs a measurement. If Alice 
measures $\sigma_\phi$ then the state
that Bob is inserting into the cloning machine is either $|0\rangle_\phi$ 
or $|1\rangle_\phi$ with probability $p=\frac12$. The overall state at
disposal of Bob, at the output of the cloning machine and without
knowing the result of Alice's measurements, is thus given by
$$
R_\smb^{(\phi)}=\frac12 \left[
|0\rangle_{\phi\phi}\langle 0| \otimes |0\rangle_{\phi\phi}\langle 0|
+ |1\rangle_{\phi\phi}\langle 1| \otimes |1\rangle_{\phi\phi}\langle 1|
\right]\,. 
$$
If Alice measures $\sigma_3$ the line of reasoning is the same and the 
state at disposal of Bob is described by the operator
$$
R_\smb^{(3)}=\frac12 \left[
|0\rangle\langle 0| \otimes |0\rangle\langle 0|
+ |1\rangle\langle 1| \otimes |1\rangle\langle 1|
\right]\,. 
$$
Since $R_\smb^{(\phi)}\neq R_\smb^{(3)}$ it would be possible for Bob to
discriminate the two states and, in turn, to infer which measurement has 
been performed by Alice just by looking at his local state.
However, this is a clear violation of the no-signaling condition 
and, in fact, the no-cloning appeared soon after as a rebuttal of 
the FLASH proposal. Or, to say the same
in other words, a perfect cloning machine would turn ensembles
corresponding to the same statistical operators into ensembles with
different statistical operators, making them distinguishable \cite{fer06}.
\par
\begin{figure}
\includegraphics[width=0.75\textwidth]{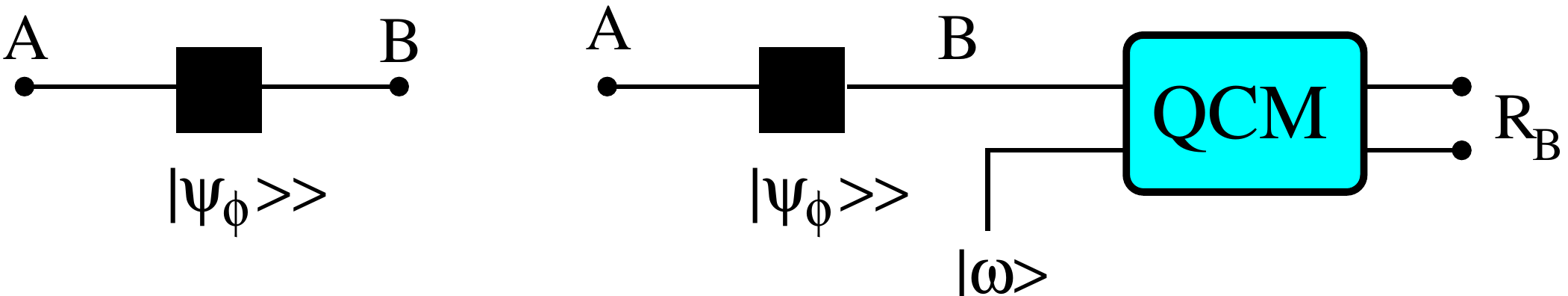}
\label{f:cl3}
\caption{
(Left): Alice and Bob share an entangled state. (Right): Schematic diagram of 
a hypothetical superluminal communication scheme exploiting a quantum
cloning machine. The use of a cloning device would allow Bob to 
discriminate which measurement has been performed by
Alice just by looking at his local state $R_\smb$.}
\end{figure}
\par
\subsection{Approximate cloning and discrimination of ensembles}
As we have seen in the previous Section, the no-signaling condition 
forbids quantum cloning and, at the same time, both exact and approximate 
discrimination of ensembles with the same statistical operator. 
On the other hand, since approximate cloning 
machine fulfilling the no-signaling condition have been suggested, one
may wonder whether this class of devices may permit approximate
discrimination of ensembles corresponding to the same statistical
operator. As we will see the answer is negative, thus showing the
full equivalence of the no-signaling condition, the no-cloning theorem
and the impossibility of discriminating ensembles with the 
same statistical operator.
\par
Assuming the reader familiar with the impossibility of perfect quantum
cloning from now on we will speak about "quantum cloning machine",
(QCM) dropping the obvious specification "approximate" (from time to time
we will forget also the term "quantum").
A generic $N\rightarrow M$ quantum cloning machine for pure states is 
a device involving a unitary operation, an ancillary system
and an apparatus system, implementing a transformation of the form
$$
|\psi\rangle^{\otimes N} \otimes |\omega\rangle^{\otimes (M-N)} \otimes
|A\rangle_\smc \stackrel{U}{\longrightarrow} |\Psi\rangle\,,
$$
and for which the partial traces 
$
\varrho_j=\Tr_{\neq j} \big[|\Psi\rangle\langle\Psi |\big]
$
possesses some similarity to the input state $|\psi\rangle$. In the
above formulas $|\psi\rangle^{\otimes N}=|\psi\rangle \otimes ... \otimes
|\psi\rangle$ denotes the $N$-fold tensor product of the state
$|\psi\rangle$ and $\Tr_{\neq j}[...]$ the partial trace over all the
systems except the $j$-th one.
In other words, one assumes to start with $N$ quantum system identically
prepared in the state $|\psi\rangle$ and aims to end up with $M$ quantum
systems prepared in some states $\varrho_j$, $j=1,...,M$, each being as close
as possible to $|\psi\rangle$ according to some figure of merit
quantifying the similarity between a pair of quantum states.
\par
The standard figure of merit, employed to assess the performances of a
cloning machine is the so-called single-clone fidelity, i.e.
$F_j=\langle\psi|\varrho_j|\psi\rangle$.  A quantum cloning
machine is said
{\em universal}: if the fidelities $F_j$ do not depend on
the input state $|\psi\rangle$, and {\em symmetric}: if $F_j=F$, 
$\forall j$.
A suitable figure of merit to globally assess a cloning machine is then 
the average of single-clone fidelities, i.e. 
$\overline{F} = \frac 1M \sum_j \int_S d\psi\, F_j (\psi)$, 
where summing over $j$ implements the average over the clones and 
the integral (being $d\psi$ a formal notation to denote a suitable
parametrization for the signals) over the set $S$ of possible input signals.
For symmetric QCM we may avoid the averaging  over the clones, and for
universal the one over the signals.
A quantum cloning machine is said
{\em optimal}: if the average fidelity is maximal,
consistently with the constraints imposed by quantum mechanics. 
The proof that a given QCM is optimal may be obtained in conjunction
with fundamental constraints. As for example an optimal QCM could be a 
device giving a fidelity which is the maximum possible without violating 
the no-signaling condition.
\par
The so-called Buzek-Hillery $1\rightarrow 2$ optimal 
cloning machine for qubit is
realized by a three-qubit unitary transformation, which acts on the signal 
basis as follows (we omit to explicitly indicate the tensor product) 
\begin{align}
|0\rangle_\sma |\omega\rangle_\smb | A\rangle_\smc
&\stackrel{U}{\longrightarrow}
\sqrt{\frac23}\, 
|0\rangle_\sma |0\rangle_\smb | 1\rangle_\smc
-\sqrt{\frac13}\, 
|\tfrac{\sigma_1}{\sqrt{2}}\rangle\rangle_{\smab} | 0\rangle_\smc
\notag\\
|1\rangle_\sma |\omega\rangle_\smb | A\rangle_\smc
&\stackrel{U}{\longrightarrow}
-\sqrt{\frac23}\,
|1\rangle_\sma |1\rangle_\smb | 0\rangle_\smc
+\sqrt{\frac13}\, 
|\tfrac{\sigma_1}{\sqrt{2}}\rangle\rangle_{\smab} | 1\rangle_\smc
\,.
\end{align}
Explicit unitaries may be written for any choice of $\omega=0,1$ and
$A=0,1$. Using the above transformations it is straightforward to see
that the generic qubit state evolves as
$$
|\psi\rangle_\sma |\omega\rangle_\smb | A\rangle_\smc
\stackrel{U}{\longrightarrow}
\sqrt{\frac23}\, 
|\psi\rangle_\sma |\psi\rangle_\smb |\psi^\ort\rangle_\smc
-\sqrt{\frac16}\, 
\Big[
|\psi\rangle_\sma |\psi^\ort\rangle_\smb 
+ |\psi^\ort\rangle_\sma |\psi\rangle_\smb 
\Big] | \psi\rangle_\smc \,,$$
where $\langle\psi|\psi^\ort\rangle=0$.
Upon taking the partial traces over the systes BC and AC respectively
one arrives at the following (identical) expression for the density
operator of the two clones
$$
\varrho_\sma = \varrho_\smb = \frac56 |\psi\rangle\langle\psi| + \frac16
|\psi^\ort\rangle\langle\psi^\ort |\:,
$$
which says that the Buzek-Hillery cloning device is symmetric and
universal and that the fidelity is given by $F=\tfrac56$.
\par
Optimality of this QCM may be proved in connection with the no-signaling 
condition, i.e. by proving that a larger fidelity would allow superluminal 
communication. To this aim let us consider the Bloch sphere representation 
of the generic input state $|\psi\rangle\langle\psi|= \frac12 \left(\id + 
\boldsymbol{r}\cdot \boldsymbol{\sigma}\right)$ and of the corresponding
clones from a universal and symmetric QCM, i.e. $\varrho_\sma=\varrho_\smb
= \frac12 \left(\id + \boldsymbol{r}^\prime\cdot\boldsymbol{\sigma}\right)$.
Since for qubits $F=\Tr[|\psi\rangle\langle\psi|\, \varrho_\sma] =
\frac12 (1+ \boldsymbol{r}\cdot\boldsymbol{r}^\prime)$ we may write
$\boldsymbol{r}^\prime=\eta \boldsymbol{r}$ where we have introduced the
so-called "shrinking factor" $0\leq\eta\leq 1$, which accounts for 
the degradation of the clones. The fidelity rewrites as $F=\frac12
(1+\eta)$ and for the Buzek-Hillery QCM we have $\eta=\frac23$.
Is this the maximum possible values? The answer is positive as it may
be proved by imposing the no-signaling condition to any scheme as the
one in the right panel of Fig. \ref{f:cl3} which employs a universal 
and symmetric QCM.  
The proof proceeds as follows: if Alice measures either $\sigma_1$ or
$\sigma_3$ and Bob use a cloning machine on his conditional state the output
states will be of the form 
$R_1=\frac12 (R_1^1 + R_1^0)$ and $R_3=\frac12 (R_3^1 + R_3^0)$.
If we impose that the $R$'s should be compatible with being the 
output of a symmetric and universal QCM i.e. such that $\Tr_\sma[R]=
\Tr_\smb[R]=\frac12(\id + \eta\, \boldsymbol{r}\cdot\boldsymbol{\sigma})$ 
and the no-signallign condition, then we end up with the condition 
$\eta\leq \frac23$. The very same threshold ensures that $R_1=R_3$, i.e.
approximate cloning {\em do not} 
turn single-qubit ensembles
corresponding to the same statistical operators into 
two-qubit distinguishable ensembles.
\section{Discrimination of seemingly equivalent ensembles}
\label{s:exa}
In this section we review some measurement schemes where discrimination
of seemingly equivalent ensembles is realized, thus leading to an
apparent contradiction.  As mentioned above, these seemingly paradoxical
situation arise when the implicit assumptions contained in introducing a
correspondence between quantum ensembles and the corresponding
single-particle statistical operator, are not satisfied. 
\subsection{Single- and many-particle density operator}
Let us consider a spin system made of $N$ particles prepared randomly 
way. The state of the system is thus described the statistical 
operator:
\begin{center}\begin{equation}\label{eq:beam}
\rho_N=\frac{1}{2^N} \sum_{j_1,\ldots, j_N} P_{j_1}\otimes\cdots\otimes
P_{j_N}\equiv\frac{1}{2^N}\mathbb{I}_{2^N}  
\end{equation}\end{center}
where $\mathbb{I}_{2^N}$ is 
the identity operator in the space of dimension $2^N$, 
$j_k=\pm 1$ $ \forall k$ and $P_{j_k}$ are projectors onto 
the subspace describing the spin of the particle $j_k$, e.g. 
eigenvectors of the Hamiltonian describing the spin of the 
single particle. Complete randomness is 
equivalent to the fact that all projectors contribute with the same weight
in the density matrix. Let us now consider the two following 
ensembles of N particles: the ensemble $E_1$contains particles 
with spin directed as $z-axis$ while $E_2$ contains particles 
with spin along $x-axis$. In formula 
\begin{align}
E_1: & |1\rangle|0\rangle|0\rangle|0\rangle\cdots\;\; random \\ 
E_2: & |\uparrow\rangle|\downarrow\rangle|\uparrow
\rangle|\uparrow\rangle\cdots\;\;random\,.
\end{align}
The two ensembles are equivalent and correspond to the same density operator
$\rho=\frac{1}{2^N} \mathbb{I}_{2^N}$.
If we look at the same systems as a collection of 
preparations of states of $m$-particle with $m<N$ then 
the density operator describing the two systems is given
by the partial trace of $\rho_N$, i.e.
$\rho_m=\frac{\mathbb{I}_{2^m}}{2^m}$. 
In particular, for ensembles of states of single-particle we have
$\rho_1=\frac{\mathbb{I}_{1}}{2}$ and thus it is not possible 
to discriminate the two ensembles neither by single-particle
measurements, nor by collective ones.
\par
Suppose now to prepare the ensemble of $N$ particles as follows: we prepare 
the first particle with spin up, the second in the state
with spin down and so on, and we separate two successive preparations by 
a time $\tau$ in order to label each particle. We then consider the
following two ensembles: $E_3$ made of particles 
with spin along $z-axis$, and $E_4$, with spin along the $x-axis$.
Explicitly 
\begin{align}
E_3:&
|1\rangle_0|0\rangle_{\tau}|1\rangle_{2\tau}|0\rangle_{3\tau}\cdots\;\;
random \\
E_4:&  |\uparrow\rangle_0|\downarrow\rangle_\tau|\uparrow\rangle_{2\tau}|
\downarrow\rangle_{3\tau}\cdots\;\; random\,.
\end{align}  
If we see these two ensembles as preparations of single-particle
states, they are equivalent, and correspond to statistical 
operator $\frac{\mathbb{I}}{2}$. However, if we look at the two
ensembles as describing preparations of many-particle states, 
then they can be discriminated by measuring some collective observable, 
e.g. $\Sigma_z = \sigma_z\otimes\mathbb{I}\otimes\cdots+
\cdots+\mathbb{I}\otimes\cdots \otimes\sigma_z$. In the following, we 
show this explicitly two- and three-particle observables. 
Eventually, we will generalize for case of N-particles. 
\par
If we look at the two ensembles $E_3$ and $E_4$ as describing a
two-particle system, we immediately realize that they correspond
to different pure states: $|1\rangle_{t}|0\rangle_{t+\tau}$ and
$|\uparrow\rangle_{t}|\downarrow\rangle_{t+\tau}$. 
The observable  $\Sigma_z$  for two particles, in the canonical basis
(eigenvectors of $\sigma_z$ and $\mathbb{I}$) corresponds to the 
matrix $\Sigma_z=\hbox{Diag}(2,0,0,-2)$ and thus we have
\begin{align}
\langle \Sigma_z\rangle_{E_3}=0 & \quad
\langle \Sigma_z^2\rangle_{E_3}=0 \\
\langle \Sigma_z\rangle_{E_4}=0 & \quad
\langle \Sigma_z^2\rangle_{E_4}=2\,,
\end{align}
making the two ensemble easily distinguishable looking at the
fluctuations of $\Sigma_z$.
\par
If we see the two ensembles $E_3$ and $E_4$ as preparations 
of three-particles system, then they correspond to the statistical
operators 
\begin{align}
\rho_{E_3} &= 
\frac12 
\Big(
|1\rangle\langle 1| \otimes |0\rangle\langle 0| \otimes |1\rangle\langle
1| + |0\rangle\langle 0| \otimes |1\rangle\langle 1| \otimes |0\rangle\langle
0|\Big) \\
\rho_{E_4}&=
\frac12 
\Big(
|\downarrow \rangle\langle \downarrow | 
\otimes |\uparrow\rangle\langle \uparrow| \otimes | \downarrow
\rangle\langle\downarrow| + |\uparrow\rangle\langle \uparrow| 
\otimes |\downarrow\rangle\langle \downarrow| \otimes |\uparrow\rangle\langle
\uparrow|\Big)
\end{align}
leading to
\begin{align}
\langle \Sigma_z\rangle_{E_3}=0 & \quad
\langle \Sigma_z^2\rangle_{E_3}=1 \\
\langle \Sigma_z\rangle_{E_4}=0 & \quad
\langle \Sigma_z^2\rangle_{E_4}=3\,,
\end{align}
More generally, looking at the two ensembles as preparations of 
$m$-particle systems, then we have 
\begin{align}
\langle \Sigma_z\rangle_{E_3}=0 & \quad
\langle \Sigma_z^2\rangle_{E_3}=1 \\
\langle \Sigma_z\rangle_{E_4}=0 & \quad
\langle \Sigma_z^2\rangle_{E_4}=m\,,
\end{align}
when $m$ is odd, and
\begin{align}
\langle \Sigma_z\rangle_{E_3}=0 & \quad
\langle \Sigma_z^2\rangle_{E_3}=0 \\
\langle \Sigma_z\rangle_{E_4}=0 & \quad
\langle \Sigma_z^2\rangle_{E_4}=m\,,
\end{align}
if $m$ is even.
\par
Looking naively at the above measurement scheme, one may conclude that
it may be used to discriminate two equivalent ensembles, even if 
$m\rightarrow\infty$. This is
definitely not the case: no single-particle measurement may be used 
to reveal differences between the two ensemble, in agreement with the
fact that they are described by the same single-particle density
operator. On the other hand, the two ensembles are prepared with 
inner correlations, i.e. correlations between subsequent preparations 
and thus they are more properly described by many-particle density 
operators, which are no longer identical, thus leaving room for
discrimination. Implications of these findings on the
conclusion that there is no quantum entanglement in the current nuclear
magnetic resonance quantum computing experiment \cite{bra99} have been discussed
\cite{lon06}.
\subsection{Finite ensembles}
Let us now address a situation involving polarization of photons and
consider two ensembles describing photons prepared either with linear or
circular polarization. In particular, we consider two equivalent
ensembles made of $N$ randomly prepared photons such that we know that
exactly $\frac12 N$ photons are prepared in one of the two states (say
linear vertical or clockwise circular) and $\frac12 N$ in the
complementary one (linear horizontal or anticlockwise circular).
If we take a photon from these ensembles we have probability $\frac12$ 
of finding a vertical/clockwise polarized photon and $\frac12$ of finding a
horizontal/anticlockwise polarized photon. Given a photon, we do not know
the state in which the photon is, but we know only the probability to
find a photon in a
state. The state of the system is thus described by the two
ensembles
\begin{align}
E5: & \Big\{\frac12, |0\rangle; \frac12, |1\rangle\Big\} \\ 
E6: & \Big\{\frac12, |+\rangle; \frac12, |-\rangle\Big\} 
\end{align}
where $|\pm\rangle= \frac{1}{\sqrt{2}} ( |0\rangle \pm |1\rangle )$. 
The two ensembles are equivalent, i.e. correspond to the same 
statistical operator $\frac12 \mathbb{I}$, and thus they cannot be
discriminated by single-particle measurements.
\par
Assume to make the following experiment: take a filter that let 
only photons polarized in $|0\rangle$ to pass. For photons prepared
according to the ensemble $E_5$ we have that $\frac12 N$ particles 
pass the filter while the rest will be blocked. For the ensemble $E_6$, 
each photon has probability $\frac12$ to pass through filter, and thus 
the number of photons that will pass is governed by the 
following statistics: 
\begin{equation}\label{eq:bernoulli}
P(N,m,p)=\binom{N}{m} p^m(1-p)^{N-m}
\end{equation}
where $m$ is the number of photons after the filter and $p$ the
survival probability for each photon, $p=\frac12$ in this case. 
In order to make the two ensembles indistinguishable, 
the number of photons after the filter should be $\frac12 N$ also for 
$E_6$, i.e., according to \eqref{eq:bernoulli} 
\begin{equation}
\label{pm}
P(N,\frac12 N,1/2)=
\binom{N}{\frac12 N}\, \left(\frac{1}{2}\right)^N 
\stackrel{N\gg1}{\simeq} \sqrt{\frac{2}{\pi N}}\,.
\end{equation}
Eq. (\ref{pm}) shows that the probability of having $\frac12 N$ 
particles after the filter decreases with $N$ and thus
the probability of discriminating 
the two ensembles increases with the number of particles.
\par
The mistake in this case is slightly harder to find. After all, 
one may say, in this case we have prepared the system randomly and 
we have not imposed any specific succession of states for the 
photons. 
However, what we have used to discriminate the two ensembles
is the fact that we know 
that {\em exactly} half of photons are in a state and half
in the other, and not just {\em half in average} as it would have been 
from the knowledge of the single-particle density matrix only.
This is equivalent to assume the presence of correlations within the
ensemble, such that the preparation of the system cannot be properly
described by single-particle statistical operator.
\section{Conclusions}
\label{s:out}
We have reviewed in details the concepts of quantum ensemble and
statistical operator, emphasizing the implicit assumptions contained in
introducing a correspondence between quantum ensembles and the
corresponding single-particle statistical operator.  We have then
discussed some issues arising when these assumptions are not satisfied,
illustrating some examples of practical where different (but equivalent)
preparations of a quantum system, i.e. different ensembles corresponding
to the same single-particle statistical operator, may be successfully
discriminated exploiting multiparticle correlations, or some a priori
knowledge about the number of particles in the ensemble. Besides, 
we have briefly discussed the connection between the no-cloning 
theorem and the impossibility of discriminating equivalent ensembles 
and shown that also approximate approximate quantum cloning 
machines cannot be used for this task. 
Overall, it appears that discrimination of equivalent preparations is 
indeed possible, but also that the involved ensembles correspond to the same
single-particle density operator but different many-particles ones.  In
other words, there are no paradoxes unless the measurement schemes are
analyzed in naive way.

\end{document}